\documentclass{ifacconf}

\makeatletter

\def\@biblabel#1{[#1]}
\makeatother
\usepackage{algorithm}
\usepackage{algpseudocode}
\usepackage{graphicx}      
\usepackage{array}      
\usepackage{booktabs}   
\usepackage{float, amsfonts,amsmath, amssymb, mathtools, xcolor, multirow, multicol}
\usepackage[numbers]{natbib}        
\AtBeginDocument{}
\makeatletter
\let\NAT@numbersfalse\NAT@numberstrue
\let\ifac@ssect\@ssect
\makeatother
\usepackage{caption}
\usepackage[hyphens]{url}
\usepackage{hyperref}
\usepackage{bookmark}
\makeatletter\let\@ssect\ifac@ssect\makeatother
\usepackage{wrapfig}

\makeatletter
\def\@overcaptionskip{0pt}
\def\section{\@startsection{section}{1}{\z@}%
  {1\p@ \@plus 1\p@ \@minus 1\p@}%
  {1\p@ \@plus 1\p@ \@minus 1\p@}%
  {\centering}}
\makeatother
\begin{document}
\begin{frontmatter}

\title{Agent-Based Evolutionary Dynamics for Mixed Autonomy Weaving Ramps}

\author[USC-CEE]{Sheryl Paul}
\author[USC-CEE]{Kexin Wang}
\author[USC-CEE]{Ruolin Li}
\author[USC-CEE]{Jyotirmoy V. Deshmukh}

\address[USC-CEE]{University of Southern California, Los Angeles, CA, USA.\\
\{sherylpa, kwang255, ruolinl, jdeshmuk\}\texttt{@usc.edu}}

\begin{abstract}

Existing models of mixed-autonomy weaving ramps characterize how altruistic connected and automated vehicles (CAVs) can improve traffic efficiency at the population level, but provide limited insight into how such behavior emerges from decentralized vehicle interactions or how it is affected by finite populations, heterogeneous preferences, and imperfect information. We develop an agent-based model of a macroscopic weaving-ramp framework in which individual vehicles adapt their lane choices using an evolutionary game-theoretic update rule and altruism-based objectives providing a microscopic interpretation of the original Wardrop model. We prove convergence of the decentralized dynamics to the unique equilibrium predicted by the macroscopic theory. Beyond reproducing aggregate equilibrium behavior, the framework enables the study of deployment-level questions that cannot be addressed by static analysis. Simulation results demonstrate close agreement with the macroscopic predictions while revealing how convergence rates, adaptation to changing traffic conditions, heterogeneous altruism levels among CAVs, and imperfect state information influence system performance and the distribution of altruistic burden across vehicles. These results provide a bridge between equilibrium traffic theory and decentralized mixed-autonomy deployment.
\end{abstract}

\end{frontmatter}

\section{Introduction}
\label{sec:intro}
Freeway weaving sections, where entering and exiting flows must
cross within a short roadway segment, remain one of the most persistent
bottlenecks on highways
\cite{cassidy2005merge}.
Recent advances in connected and automated vehicles (CAVs) create an
opportunity to influence traffic conditions through lane-choice
behavior~\cite{stern2018dissipation}. Unlike human-driven vehicles (HDVs), whose decisions emerge from
individual preferences, CAVs can be programmed to account for broader
traffic objectives while remaining decentralized.

Vehicle-level traffic models examine how individual interactions and control
decisions affect lane changing, stability, and throughput, while Social Value Orientation (SVO)-based
models describe how CAVs balance private and collective costs
\cite{talebpour2016influence,schwarting2019socialbehavior}. A complementary
evolutionary-game-theoretic literature studies how strategy shares change over
time, from evolutionarily stable strategies and replicator dynamics to
traffic-assignment processes that converge to Wardrop equilibria
\cite{maynardsmith1973,taylor1978,smith1984}.

At the macroscopic scale, recent work has developed a game-theoretic model of
lane choice at weaving ramps
\cite{he2025staybypass,wang2025altruism}. The model represents HDVs through a
Wardrop equilibrium and incorporates SVO into CAV objectives, allowing CAVs to
trade off private delay against aggregate congestion. It characterizes when
altruistic CAVs improve traffic efficiency, derives closed-form penetration
thresholds, and quantifies the achievable reduction in social delay. Thus, the
mixed-autonomy weaving problem is analytically characterized at equilibrium.

The two modeling perspectives nevertheless leave a complementary gap. The macroscopic theory characterizes equilibrium lane-choice fractions but does
not provide a decentralized vehicle-level mechanism for reaching them.
Equilibrium existence also does not guarantee convergence under local learning;
adjustment may be slow, oscillatory, or unsuccessful
\cite{hart2003uncoupled,sandholm2010}. Moreover, the static formulation cannot
capture adaptation to changing demand, determine whether finite heterogeneous
fleets reproduce mean-field predictions, or explain how the altruistic burden
is distributed across vehicles. It also assumes exact traffic information,
whereas deployed CAVs rely on noisy local observations and imperfect state
estimates. On the other hand, the aforementioned vehicle-level
models were not designed to reproduce the calibrated cost
structure and penetration thresholds of such a weaving-ramp model that incorporates CAVs as well as HDVs.
They therefore do not establish whether finite-agent adaptation recovers its
aggregate predictions under realistic deployment constraints.

These limitations motivate the central question:
\emph{Can the equilibrium prescriptions of the macroscopic theory be realized
by a finite population of decentralized vehicles, and if so, what behavior
emerges under adaptation, heterogeneous altruism, and imperfect information?}
This paper answers the question by treating the calibrated macroscopic model of \cite{he2025staybypass,wang2025altruism} as a specification and developing a new agent-based model (ABM) of it. The equilibrium cost structure and penetration benchmarks are inherited from those papers; the vehicle-level dynamics, convergence analysis, and deployment studies below are novel. 

\textbf{Contributions:} \textit{(i)}  We develop a decentralized realization of the macroscopic weaving-ramp model using replicator-based adaptation from evolutionary game theory \cite{weibull1995,sandholm2010}. We show that the resulting lane-choice game is an exact potential game and that the induced dynamics converge globally to the unique Wardrop equilibrium.
\textit{(ii)} Using the resulting agent-based model, we characterize deployment phenomena that are absent from the equilibrium theory, including convergence regimes, demand-tracking limits, and the distribution of costs and benefits across altruistic and selfish vehicles.
\textit{(iii)} We extend the framework to heterogeneous altruism and partial observability and identify when aggregate behavior is summarized by an effective altruism angle and analyze the impact of noise. 
\textit{(iv)} We show that belief sharing improves performance under stationary demand, but can be conservative under abrupt changes. To address this, we introduce a forgetting-factor update that balances accuracy and adaptability.

\section{Preliminaries}
\label{sec:prelims}

\subsection{Population-game Model}
We briefly introduce the population-game concepts used in the paper. Consider a population of agents choosing between two pure strategies,
$S$ (stay) and $B$ (bypass), and let $x\in[0,1]$ denote the fraction of agents selecting~$S$, so that $1-x$ select~$B$.

A two-strategy \emph{population game} is defined by a pair of cost functions
$J^S, J^B : [0,1]\to\mathbb{R}$, where $J^S(x)$ and $J^B(x)$ denote
the costs incurred by agents choosing strategies $S$ and $B$,
respectively, at population state $x$.
The mean population cost is
    $\bar{J}(x) = x J^S(x) + (1-x) J^B(x)$.

A state $x^\star\in[0,1]$ is a \emph{Wardrop equilibrium} if no agent
can reduce its cost by unilaterally switching strategy:
$x^\star \cdot \bigl(J^S(x^\star) - J^B(x^\star)\bigr) \leq 0; \ \  
    (1-x^\star) \cdot \bigl(J^B(x^\star) - J^S(x^\star)\bigr) \leq 0.$
For an interior equilibrium $x^\star\in(0,1)$, this reduces to $J^S(x^\star) = J^B(x^\star)$.
Wardrop equilibrium is the user-equilibrium concept from
transportation theory and coincides with the symmetric Nash equilibrium
of the corresponding population game \cite{wardrop1952,sandholm2010}.

The continuous-time replicator dynamics associated with a population
game $(J^S, J^B)$ are
\begin{equation}
    \dot{x} = x(1-x)\bigl(J^B(x) - J^S(x)\bigr).
    \label{eq:replicator}
\end{equation}
Replicator dynamics model decentralized adaptation: strategies with
lower cost increase their population share over time.
Any interior stationary point satisfies $J^S(x^\star)=J^B(x^\star)$
and therefore coincides with an interior Wardrop equilibrium.
In the agent-based model, vehicles update their lane-choice
probabilities using a discrete-time replicator rule; as population
size grows and the step size shrinks, the stochastic process converges
to Eq.~\eqref{eq:replicator} \cite{benaim2003,sandholm2010},
so the replicator ODE serves as its analytical approximation.

A population game is an \emph{exact potential game} if there exists a
continuously differentiable function $\Phi:[0,1]\to\mathbb{R}$
satisfying
$\tfrac{d\Phi}{dx} = J^B(x) - J^S(x)
    , \ \  \forall\, x\in[0,1].$

In an exact potential game the stationary points of the replicator
dynamics coincide with those of the potential, which can be used
directly as a Lyapunov function to establish stability
\cite{monderer1996,rosenthal1973,sandholm2010}.

Let $x^S$ and $x^B$ denote the fractions of vehicles choosing
the stay and bypass strategies, ($x^S+x^B=1$). The agent's private costs are
\begin{equation}
    J^{\mathrm{self},S}=K_1^Sx^S+B_1^S, \ \ 
J^{\mathrm{self},B}=K_1^Bx^B+B_1^B
 \label{eq:lane_costs}
\end{equation}
The $K_1$ and $B_1$ coefficients capture the calibrated
traversing and merging costs; their detailed expressions are provided in Section~\ref{sec:macro}. The aggregate social cost is the sum of the costs incurred
by all traffic classes, weighted by their respective flow fractions and is denoted by $J^{\mathrm{soc}}$. 

Social Value Orientation (SVO) represents an agent’s relative concern for private and collective outcomes by an altruism angle $\theta\in[0,\pi/2]$. For a lane choice $a \in \{S, B\}$, the agent minimizes the
effective cost
\begin{equation}
    \tilde{J}^a(\theta)
    = \cos(\theta)\,J^{\mathrm{self},a}
    + \sin(\theta)\,({\partial J^{\mathrm{soc}}}/{\partial x^a}),
    \label{eq:svo}
\end{equation}
where $J^{\mathrm{self},a}$ is the agent's private travel cost and $\frac{\partial J^{\mathrm{soc}}}{\partial x^a}$ is the marginal social cost of its lane choice, and the sine–cosine weights are the standard unit-circle parameterization of SVO~\cite{vanlange1999,murphy2011,liebrand1986}.
The boundary cases are $\theta=0$ (fully selfish) and $\theta=\pi/2$
(fully altruistic).

\subsection{Weaving Ramp Model}
\label{sec:macro}

We now present the macroscopic weaving-ramp model
\cite{he2025staybypass,wang2025altruism} and state
 equilibrium predictions that the agent-based model of
Sec.~\ref{sec:micro} will be shown to match.

\textit{Setup and Parameters.}
We consider a single highway weaving section with two lanes.
Vehicles are partitioned into three movement classes by their origin
and destination: \emph{entering} vehicles, which join the highway at
the on-ramp; \emph{exiting} vehicles, which leave at the off-ramp; and
\emph{through} vehicles, which traverse the section without changing
lanes.
Each \emph{through} vehicle in the mainline lane adjacent to the ramps additionally selects one of two lane-choice strategies: \emph{stay} ($S$), remaining in lane~1, or
\emph{bypass} ($B$), shifting to lane~2 to avoid the merge conflict.
The model is one-shot at the equilibrium-flow scale: lane choices
within the weaving region are resolved into stationary flow ratios,
and intra-section spatial dynamics are abstracted into the lane-cost
functions below.

Let $n=(n_{\mathrm{enter}},n_{\mathrm{exit}},n_2)$ denote the
normalized demand composition, where they denote the fractions of entering, exiting, and
through vehicles, respectively.
Let $x\in[0,1]$ denote the fraction of lane 1 through vehicles choosing $S$, so that $1-x$ choose $B$.
Let $p\in[0,1]$ denote the CAV penetration rate; the remaining $1-p$
fraction are HDVs.
The cost functions depend on four coefficients $(C_1^\text{t},C_1^\text{m},C_2^\text{t},C_2^\text{m})$ from the optimization method in~\cite{li2019extended} and six coefficients
$(\alpha,\beta,\gamma,\delta,\omega,\rho)$, calibrated against
microsimulation in~\cite{he2025staybypass}.\footnote{
$(C_1^\text{t},C_1^\text{m},C_2^\text{t},C_2^\text{m})=(1.000,\,1.000,1.000,\,1.000)$} 
\footnote{$(\alpha,\beta,\gamma,\delta,\omega,\rho)
=(1.255,\,1.138,\,2.384,\,3.094,\,1.000,\,1.000).$}

\textit{Lane Costs.}
The travel delay costs of $S$ and $B$ are
affine in the lane-choice state $x$, as seen in ~\eqref{eq:lane_costs}:
with
$K_1^{S} = C_1^\text{t} \alpha + C_1^\text{m} (\omega\, n_{\mathrm{exit}} + n_{\mathrm{enter}}),
B_1^{S} = C_1^\text{t} (\beta\, n_{\mathrm{exit}} + n_{\mathrm{enter}}),
K_1^{B} = C_2^\text{t} \gamma + C_2^\text{m} (\rho\, n_2 + \delta\, n_{\mathrm{exit}}),
B_1^{B} = C_2^\text{t} n_2.$

The aggregate social cost $J^{\mathrm{soc}}(x,n)$ is the total social cost experienced by all vehicles in the weaving section, obtained by summing each class's individual cost weighted by its flow: (i) steadfast and bypassing vehicles in lane~1, (ii) through vehicles in lane~2, and (iii) the entering and exiting vehicles and can be given as: $J^{\mathrm{soc}}
={} x^S(J^{\mathrm{self,s}})
   + x^B(J^{\mathrm{self,b}}) + n_2(K_2^Sx^B+B_2^S) 
 + n_{\mathrm{exit}}
    (K_2^{\mathrm{exit}}x^S+B_2^{\mathrm{exit}}) + n_{\mathrm{enter}}
    (K_0^{\mathrm{enter}}x^S+B_0^{\mathrm{enter}}).$
The strategy-specific marginal
social costs are
\begin{equation}
M^a(x,n) := \bigl.\partial J^{\mathrm{soc}}/\partial x^a\bigr|_{(x^S,x^B)=(x,1-x)},
\  a\in\{S,B\}.
\label{eq:msc}
\end{equation}
For the calibrated affine lane-cost model, $J^{\mathrm{soc}}$ is quadratic in $x$, and $M^s$ and $M^b$ are affine.

\textit{Equilibrium Benchmarks.}
Two benchmark lane-choice states organize the analysis (i) \emph{HDV Wardrop point} $\xi(n)$, the equilibrium
of the pure-HDV population ($p=0$), obtained by setting
$J_1^{S}=J_1^{B}$ in~\eqref{eq:lane_costs}, and (ii) the \emph{social optimum} $\Gamma(n)$, the steadfast
fraction that minimizes the aggregate cost:
\begin{equation}
\label{eq:phi}
{\xi(n)} \;=\; \tfrac{K_1^{B}+B_1^{B}-B_1^{S}}{K_1^{S}+K_1^{B}} ;\  \ \ \Gamma(n) \;=\; \arg\min_{x\in[0,1]} J^{\mathrm{soc}}(x,n).
\end{equation}
The admissible configurations of interest are those in which the
selfish equilibrium and the social optimum are distinct, i.e.\
$0<\xi(n)<\Gamma(n)<1$; under this
condition altruistic CAVs can steer the population from
$\xi$ toward $\Gamma$.\footnote{This ordering applies to operational regimes in which selfish traffic supplies fewer steadfast vehicles and increasing the weight on marginal social cost moves the equilibrium from $\xi$ toward $\Gamma$. }

\textit{Three-Phase Penetration Structure.}
The aggregate impact of altruistic CAVs as the CAV penetration $p$
varies is established in~\cite{wang2025altruism}: for any
admissible configuration, $J^\text{soc}$ at the population
equilibrium follows a three-phase structure in $p$.
(i) At low penetration $p\in[0,p_1]$, $J^\text{soc}$ remains at the
HDV-only reference level; (ii) over an
intermediate band $p\in(p_1,p_2]$, it decreases strictly in $p$; and (iii)
for $p\in(p_2,1]$, it saturates at the social-optimum level.
\footnote{The two transition points are:
$p_1 \;=\; \tfrac{{K_1^{B}+B_1^{B}-B_1^{S}}}{K_1^{S}+K_1^{B}}, \text{ and }
p_2 \;=\;\tfrac{(2K_1^{B}+B_1^{B}-B_1^{S}-n_\text{exit}K_2^\text{exit}-n_\text{enter}K_0^\text{enter}+n_2K_2^\text{s})}{(2\,(K_1^{S}+K_1^{B}))}$
where $K_2^{\mathrm{exit}}$ and $K_0^{\mathrm{enter}}$
are the $x^S$-coefficients in exiting and entering costs, while $K_2^s$ is the $x^B$-coefficient in Lane-2 through costs.}
This three-phase structure is the central macroscopic prediction that
the ABM of Sec.~\ref{sec:micro} will be shown to
reproduce.

\textit{SVO Effective Cost.}
Specializing the definition of SVO~\eqref{eq:svo} to the present setting, a CAV
with altruism angle $\theta\in[0,\pi/2]$ minimizes the
effective cost:
\begin{equation}
\label{eq:svo_weaving}
\tilde{J}^{(\cdot)}(x,n;\theta)
\;=\;
\cos(\theta)\,J_1^{(\cdot)}(x)
\;+\;
\sin(\theta)\,{M^{(\cdot)}(x,n)},
\end{equation}
for $(\cdot)\in\{S,B\}$, with $J_1^{(\cdot)}$ given
by~\eqref{eq:lane_costs} and $M$ by~\eqref{eq:msc}.
The angle $\theta$ is a single broadcastable parameter: HDVs use
$\theta=0$, while the CAV fleet adopts a value assigned by the
operator.
The induced mixed-population equilibrium interpolates between $\xi$
(at $\theta=0$) and $\Gamma$ (at $\theta=\pi/2$), and is the object
the ABM of the next section converges to.

\label{sec:abm}
\section{Agent-Based Model and Theory}
\label{sec:micro}
This section constructs the agent-based model (ABM) that realizes the aggregate
statements over population fractions of the macroscopic model
at the level of individual vehicles.
The ABM uses the lane costs~\eqref{eq:lane_costs}, the social-cost
machinery, and the SVO effective cost~\eqref{eq:svo_weaving} 
as specified in Sec.~\ref{sec:macro}.

\textit{Agents and Lane-Choice Dynamics.}
\label{subsec:abm}
We consider a population of $N$ vehicles indexed by
$i\in\{1,\dots,N\}$.
Each vehicle is characterized by its (fixed) movement class
$c_i\in\{\mathrm{enter},\mathrm{exit},\mathrm{through}\}$, its (fixed)
altruism angle $\theta_i\in[0,\pi/2]$, and its (time-varying)
lane-choice probability $\pi_i(t)\in[0,1]$ for selecting the steadfast
strategy $S$.
\begin{equation}
    x(t) \;=\; \tfrac{1}{N}\sum_{i=1}^{N}\pi_i(t)
    \label{eq:abm-aggregate}
\end{equation}
HDVs have $\theta_i=0$; CAVs have $\theta_i>0$, with the value broadcast  by the fleet operator. 

The aggregate steadfast
  fraction is given by \eqref{eq:abm-aggregate}.
Each vehicle updates its lane-choice probability by a discrete-time
replicator rule applied to its own SVO cost:
\begin{flalign}
\label{eq:abm_update}
&\scalebox{0.87}{$\displaystyle
\pi_i(t+1) = \pi_i(t) + \eta\,\pi_i(t)\bigl(1-\pi_i(t)\bigr)
\bigl[\tilde{J}_i^{B}\bigl(x(t)\bigr)-\tilde{J}_i^{S}\bigl(x(t)\bigr)\bigr]
$}
\end{flalign}
where $\tilde{J}_i^{(\cdot)}$ is the effective
cost~\eqref{eq:svo_weaving} evaluated at the agent's own angle
$\theta_i$, and $\eta>0$ is a step size.\footnote{We assume that $\eta$ is small enough that
$\eta\lvert\tilde J_i^B-\tilde J_i^S\rvert\leq 1$, ensuring that
$\pi_i(t)\in[0,1]$ is preserved by the update.}
The rule is decentralized: each vehicle uses its own altruism
angle, the realized lane costs, and the demand composition needed to
form the marginal social cost $M$.
The operator's only intervention is the choice of $\theta$ broadcast
to the CAV fleet.

\textit{Mean-Field Approximation and Convergence.}
\label{subsec:meanfield}

We derive a scalar continuous-time approximation using the rescaled
time $\tau=\eta t$. This follows standard stochastic-approximation
arguments~\cite{benaim2003}, and its agreement with the aggregate
agent-based dynamics is assessed numerically.

Let
\(
\Delta_0(x):=J_1^B(x)-J_1^S(x)
\)
denote the cost advantage of choosing the steadfast strategy for an
HDV. For a CAV with altruism angle $\theta$, the corresponding
effective-cost difference is
\(
\Delta_\theta(x)
=
\cos\theta\,\Delta_0(x)
+
\sin\theta\,[M^B(x)-M^S(x)].
\)
Suppose that a fraction $1-p$ of the strategic vehicles are HDVs and
a fraction $p$ are CAVs with SVO angle $\theta$. Their
population-average incentive is then:
\begin{equation}
    \Delta^{\mathrm{pop}}(x)
=
(1-p)\Delta_0(x)+p\Delta_\theta(x)
\label{eq:delta-pop}
\end{equation}
From ~\eqref{eq:abm-aggregate} and ~\eqref{eq:abm_update}, the exact aggregate change generated by the individual probability
updates is
\(
\frac{x(t+1)-x(t)}{\eta}
=
\frac{1}{N}\sum_{i=1}^{N}
\pi_i(t)\bigl(1-\pi_i(t)\bigr)
\Delta_{\theta_i}(x(t)).
\)
To obtain a scalar approximation, we assume that the individual
probabilities remain concentrated around their population average,
giving:
\(
({1}/{N})\sum_{i=1}^{N}
\pi_i(1-\pi_i)\Delta_{\theta_i}(x)
\approx
x(1-x)\Delta^{\mathrm{pop}}(x).
\) 
\begin{equation}\label{eq:meanfield}
  \dot{x} = x(1-x)\,\Delta^{\mathrm{pop}}(x),
\end{equation}
After introducing the rescaled time $\tau=\eta t$, the resulting
continuous-time approximation is \eqref{eq:meanfield}.
Intuitively, $x(1-x)$ measures the population's responsiveness,
whereas $\Delta^{\mathrm{pop}}(x)$ measures its average incentive to
choose the steadfast strategy.

\begin{theorem}[\textbf{Aggregate potential and convergence}]
\label{thm:convergence}
Under the homogenization closure, suppose $K^S+K^B>0$ and $a>0$.\footnote{With positive coefficients $(\alpha,\beta,\gamma,\delta,\omega,\rho)$
and nonnegative demand fractions, $J_{\mathrm{soc}}$ is strictly convex in $x$
with $a
= 4(\alpha+\gamma+n_{\mathrm{enter}}+(\omega+\delta)\,n_{\mathrm{exit}}
+\rho\,n_2) > 0$, identical for the lane-1 and full social-cost forms
since the additional movement terms are linear in $x$. Likewise
$K^S+K^B = \alpha+\gamma+n_{\mathrm{enter}}+(\omega+\delta)n_{\mathrm{exit}}
+\rho n_2 > 0$.}
Then, for every $p\in[0,1]$ and $\theta\in[0,\pi/2]$:\\
\emph{(i)} $\Phi^{\mathrm{pop}}$ is a strictly concave exact potential
for the aggregate mean-field game and therefore has a unique maximizer
$x^\star\in[0,1]$.\\
\emph{(ii)} If $x^\star\in(0,1)$, it is the unique interior aggregate
Wardrop equilibrium and the unique interior stationary point
of~\eqref{eq:meanfield}. It is globally asymptotically stable on
$(0,1)$ with Lyapunov function
$V(x)=\Phi^{\mathrm{pop}}(x^\star)-\Phi^{\mathrm{pop}}(x)$.\\
\emph{(iii)} Under $u=-J$, the interior maximizer is the unique
symmetric Nash equilibrium of the representative aggregate game and
is evolutionarily stable.\footnote{The model is written in costs (which agents minimize), whereas potential-game and evolutionary-stability statements are conventionally written in payoffs (which agents maximize). We use the payoff convention $u = -J$. Under this, lower-cost strategy corresponds to higher potential, and  convexity of social cost to  concavity of the potential.}
\end{theorem}

\begin{proof}[Sketch]
By definition, $\Phi_0'=\Delta_0$ and $\Phi_\theta'=\Delta_\theta$, we have
$(\Phi^{\mathrm{pop}})'=\Delta^{\mathrm{pop}}$, proving the exact
potential property. Moreover, from ~\eqref{eq:delta-pop}
\(
(\Phi^{\mathrm{pop}})''
= (\Delta^{\mathrm{pop}})' =
-(1-p)(K_1^S+K_1^B)
-p\bigl[(K_1^S+K_1^B)\cos\theta+a\sin\theta\bigr]<0,
\)
so the potential is strictly concave and has a unique maximizer.
When this maximizer is interior,
$(\Phi^{\mathrm{pop}})'(x^\star)=0$ is both the aggregate Wardrop
condition and the stationary condition for~\eqref{eq:meanfield}.
Furthermore,
$\dot V=-x(1-x)[(\Phi^{\mathrm{pop}})'(x)]^2<0$ for every
$x\in(0,1)\setminus\{x^\star\}$, proving global asymptotic stability.
The Nash and ESS conclusions follow from the strict potential maximum
\cite{monderer1996,sandholm2010,weibull1995}.
\end{proof}
\textit{Heterogeneous Fleets.}
\label{subsec:hetero_theory}
We next allow each CAV's fixed altruism angle to be drawn from a
distribution $F$ on $[0,\pi/2]$. Under the same aggregate
homogenization closure, the expected cost difference becomes
\begin{align}
\Delta_F^{\mathrm{pop}}(x)
={}&(1-p)\Delta_0(x)
+p\,\mathbb E_F[\cos\theta]\,\Delta_0(x) \notag\\
&+p\,\mathbb E_F[\sin\theta]\,[M^B(x)-M^S(x)].
\label{eq:heterogeneous_drift}
\end{align}
Thus, at the aggregate mean-field level, $F$ enters the dynamics only
through $\mathbb E_F[\cos\theta]$ and
$\mathbb E_F[\sin\theta]$.
Define the effective angle and coherence factor by
\(
\bar\theta_{\mathrm{eff}}
=
\operatorname{atan2}
\left(
\mathbb E_F[\sin\theta],
\mathbb E_F[\cos\theta]
\right),
 \)
 \(
\rho_F
=
\sqrt{
\mathbb E_F[\cos\theta]^2+
\mathbb E_F[\sin\theta]^2
}.
\)
Then
\(
\mathbb E_F[\cos\theta]
=
\rho_F\cos\bar\theta_{\mathrm{eff}},
\ \ 
\mathbb E_F[\sin\theta]
=
\rho_F\sin\bar\theta_{\mathrm{eff}},
\)
and~\eqref{eq:heterogeneous_drift} can be written as
\(
\Delta_F^{\mathrm{pop}}
=
(1-p)\Delta_0
+p\rho_F\Delta_{\bar\theta_{\mathrm{eff}}}.
\)
Equivalently, let
\(
\lambda_F:=1-p+p\rho_F,
\ \ 
p_{\mathrm{eff}}
:=
\frac{p\rho_F}{1-p+p\rho_F}.
\)
The heterogeneous drift then factors as
\(
\Delta_F^{\mathrm{pop}}
=
\lambda_F
[
(1-p_{\mathrm{eff}})\Delta_0
+p_{\mathrm{eff}}
\Delta_{\bar\theta_{\mathrm{eff}}}
].
\)
Hence its equilibrium and phase portrait coincide with those of the
homogeneous aggregate model having angle
$\bar\theta_{\mathrm{eff}}$ and penetration
$p_{\mathrm{eff}}$; the positive factor $\lambda_F$ changes only the
time scale.\\
Because $\rho_F\leq1$, with equality if and only if all CAVs share the
same altruism angle, heterogeneity reduces the fleet's effective
aggregate influence. The conclusions of
Thm.~\ref{thm:convergence} therefore continue to hold for the
heterogeneous aggregate mean-field model after replacing
$(p,\theta)$ by
$(p_{\mathrm{eff}},\bar\theta_{\mathrm{eff}})$. \\
This equivalence concerns the expected aggregate dynamics. Finite
fleets with persistent agent-specific altruism may differ because
lane-choice probabilities can become correlated with behavioral type
and because of finite-population sampling variability.

\textit{Partial Observability and Cooperative Belief Sharing.}
\label{subsec:partial_theory}
The SVO effective cost in~\eqref{eq:svo_weaving} depends on the demand
composition
\(n=(n_{\mathrm{enter}},n_{\mathrm{exit}},n_2),
\) which determines both the private lane costs and the marginal social
costs. The macroscopic model assumes that $n$ is known exactly. In
practice, a deployed CAV has access only to noisy local observations of
the traffic state. We model the measurement of vehicle $i$ at
time $t$ as: $y_{i,t}
=
n+\varepsilon_{i,t},
\ 
\varepsilon_{i,t}
\sim
\mathcal N(0,\sigma_{\mathrm{obs}}^2 I)$,
with observations independent across vehicles and time.

Each vehicle replaces $n$ in~\eqref{eq:svo_weaving} by an estimate
$\hat n_{i,t}$ formed from available information. We compare two
estimators: 
\textit{(i) Memoryless estimator.}
The vehicle uses only its most recent observation, $\hat n_{i,t}^{\mathrm{mem}}
=
y_{i,t}$.
This models a one-shot lane decision based solely on the current local
measurement. \textit{(ii) Cooperative belief-sharing estimator.}
Vehicles maintain and communicate beliefs regarding the traffic
composition. Let $\mu_{i,t}$ denote vehicle $i$'s belief and let
$\tau_{i,t}=1/\sigma_{i,t}^{2}$ denote its precision. Whenever a new
observation $y_{j,t+1}$ is received from another vehicle or roadside
infrastructure, the belief is updated according to
$\tau_{i,t+1}=\tau_{i,t} + \tfrac{1}{\sigma_{\mathrm{obs}}^2},
\ \ 
\mu_{i,t+1} =\mu_{i,t}+
\tfrac{1/\sigma_{\mathrm{obs}}^2}
{\tau_{i,t+1}}
(y_{j,t+1}-\mu_{i,t}),$ and the estimate used in the SVO cost is
$\hat n_{i,t}^{\mathrm{share}}
=
\mu_{i,t}$.
This update is the standard Bayesian fusion rule for
independent Gaussian measurements and can be interpreted as a simple
distributed information-sharing protocol. The next result quantifies the benefit of information-sharing: combining independent
observations keeps the demand estimate unbiased while reducing variance.
\begin{lemma}[\textbf{Estimator bias and variance}]
\label{lem:estimators}
Both estimators are unbiased,
$\mathbb E[\hat n_{i,t}^{\mathrm{mem}}]=\mathbb E[\hat n_{i,t}^{\mathrm{share}}]=n$,
with $\mathrm{Var}[\hat n_{i,t}^{\mathrm{mem}}]=\sigma_{\mathrm{obs}}^2 I$
and $\mathrm{Var}[\hat n_{i,t}^{\mathrm{share}}]=(\sigma_{\mathrm{obs}}^2/t)\,I$.
\end{lemma}
\begin{proof}
The memoryless estimate is $\hat n_{i,t}^{\mathrm{mem}}=y_{i,t}=n+\varepsilon_{i,t}$,
giving mean $n$ and variance $\sigma_{\mathrm{obs}}^2 I$. Bayesian fusion of $t$
independent Gaussian observations has posterior mean equal to the sample average
$\hat n_{i,t}^{\mathrm{share}}=\tfrac1t\sum_{k=1}^{t}y_k$, hence mean $n$ and
variance $\sigma_{\mathrm{obs}}^2/t\,I$.
\end{proof}
Because the effective-cost difference depends linearly on the estimated demand,
both unbiased estimators recover the full-information equilibrium on average.
They differ only in how much noise they introduce into the updates.
\begin{corollary}[Equilibrium is preserved]
\label{cor:partial_eq}
The effective-cost difference $\Delta(x,n)=\widetilde J^B-\widetilde J^S$ is affine
in $n$, so $\Delta(x,n)=a(x)^\top n+b(x)$ and, by Lemma~\ref{lem:estimators},
$\mathbb E[\Delta(x,\hat n_{i,t})\mid x]=a(x)^\top n+b(x)=\Delta(x,n)$ for either
estimator. The expected replicator vector field and the equilibrium location are
thus identical to the full-information case; the estimators differ only in
variance, which is $O(\sigma_{\mathrm{obs}}^2)$ for the memoryless rule and decays
as $1/t$ under belief sharing.
\end{corollary}

\section{Experimental Results}

We evaluate the agent-based realization of Sec.~\ref{sec:micro} using
the parameters of~\cite{wang2025altruism}, with $N=1000$ agents and step
size $\eta=0.08$.

\textbf{Macroscopic-equilibrium validation.}
For the HDV-only case ($p=0$), Fig.~\ref{fig:abm_validation} compares
the ABM steady-state bypass fraction with the closed-form Wardrop
equilibrium $\xi(n)$ and an independent SUMO simulation. All three agree
within $1$--$2\%$ across the tested demand conditions.

\begin{figure}[t]
\centering
\includegraphics[width=0.95\linewidth]{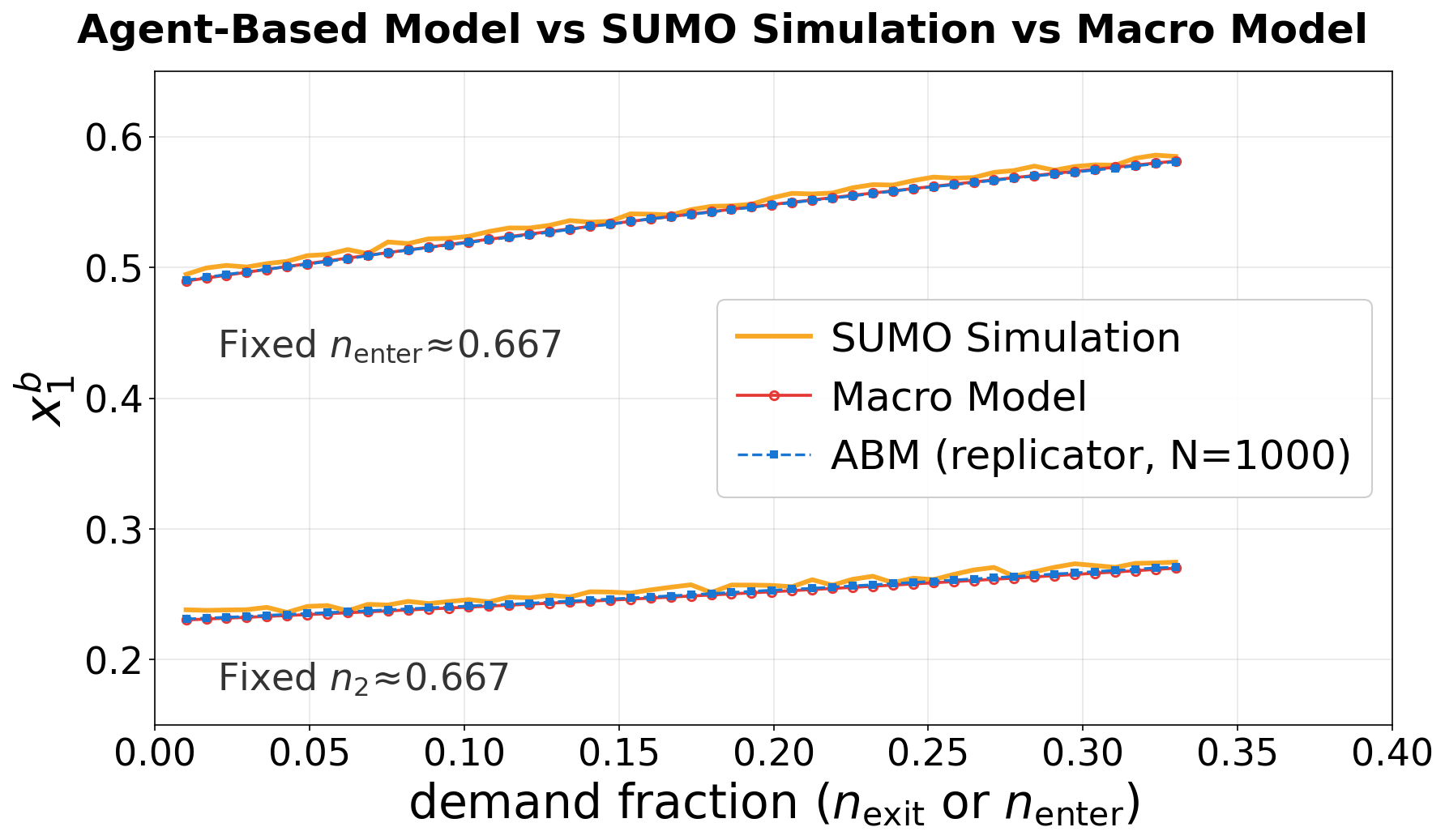}
\caption{HDV-only ABM, closed-form Wardrop equilibrium, and independent
SUMO simulation.}
\label{fig:abm_validation}
\end{figure}

For mixed traffic, Fig.~\ref{fig:cav_analytic_vs_abm} compares the
analytic equilibrium and ABM results over CAV penetration $p$ and
altruism angle $\theta$. Both recover the three regimes
of $p_1$ and $p_2$: an HDV baseline for $p\le p_1$, a decreasing
transition, and a $\theta$-dependent floor for $p\ge p_2$.

\begin{figure}[t]
\centering
\includegraphics[width=\linewidth]{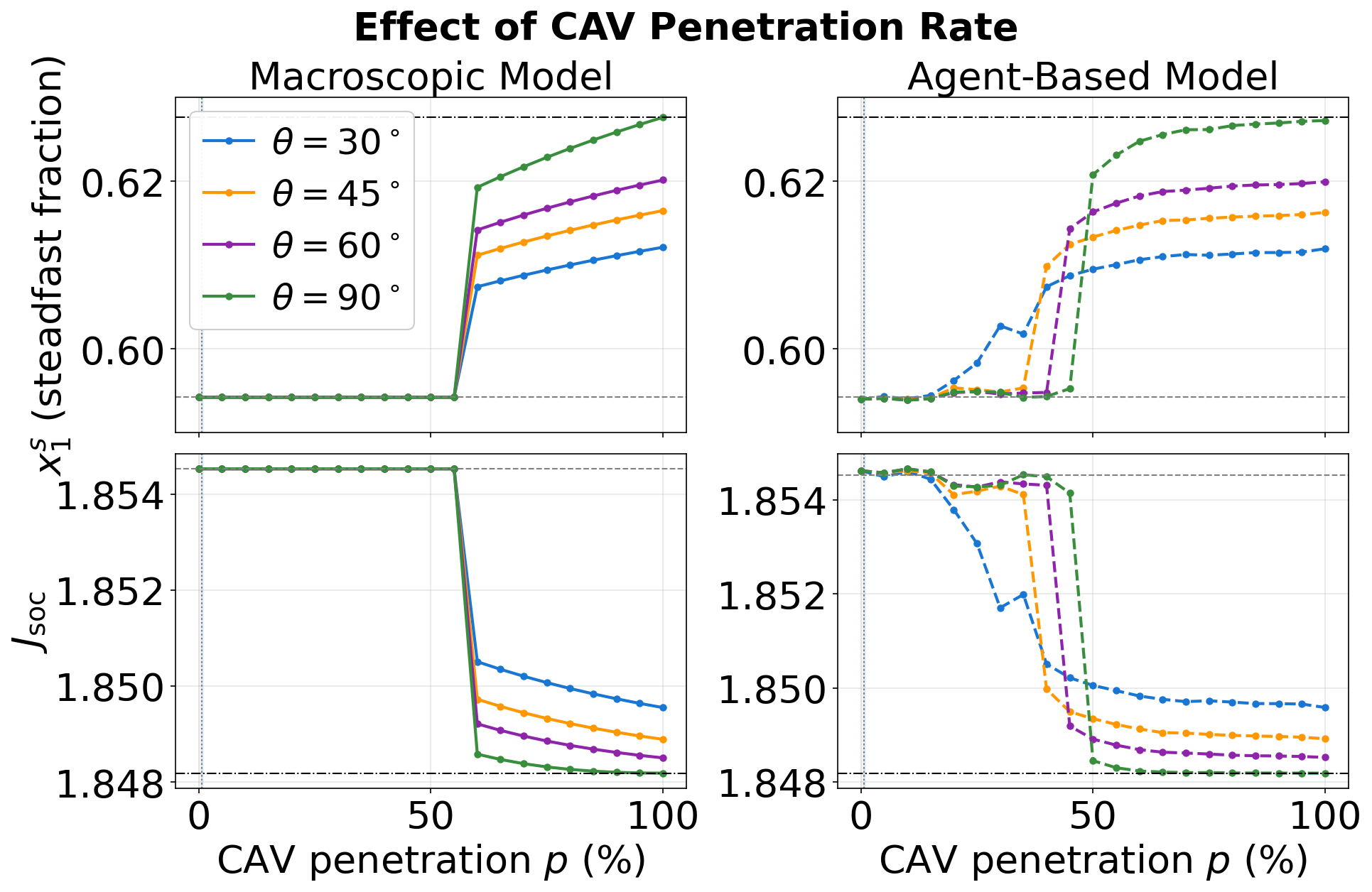}
\caption{Analytic mixed-population equilibrium (left) and ABM results
(right), including thresholds $p_1$ and $p_2$.}
\label{fig:cav_analytic_vs_abm}
\end{figure}

\textbf{Heterogeneous fleets.}
We test four homogeneous profiles and five mixtures over
$\theta\in\{\pi/6,\pi/4,\pi/3,\pi/2\}$. As shown in
Fig.~\ref{fig:weighted_theta}, all profiles retain the three-regime
structure. Profiles with the same
$\bar\theta_{\mathrm{eff}}$ produce nearly identical equilibria, while
greater heterogeneity primarily affects transient convergence through
the coherence factor $\rho_F$.

\begin{figure}[t]
\centering
\includegraphics[width=0.93\linewidth]{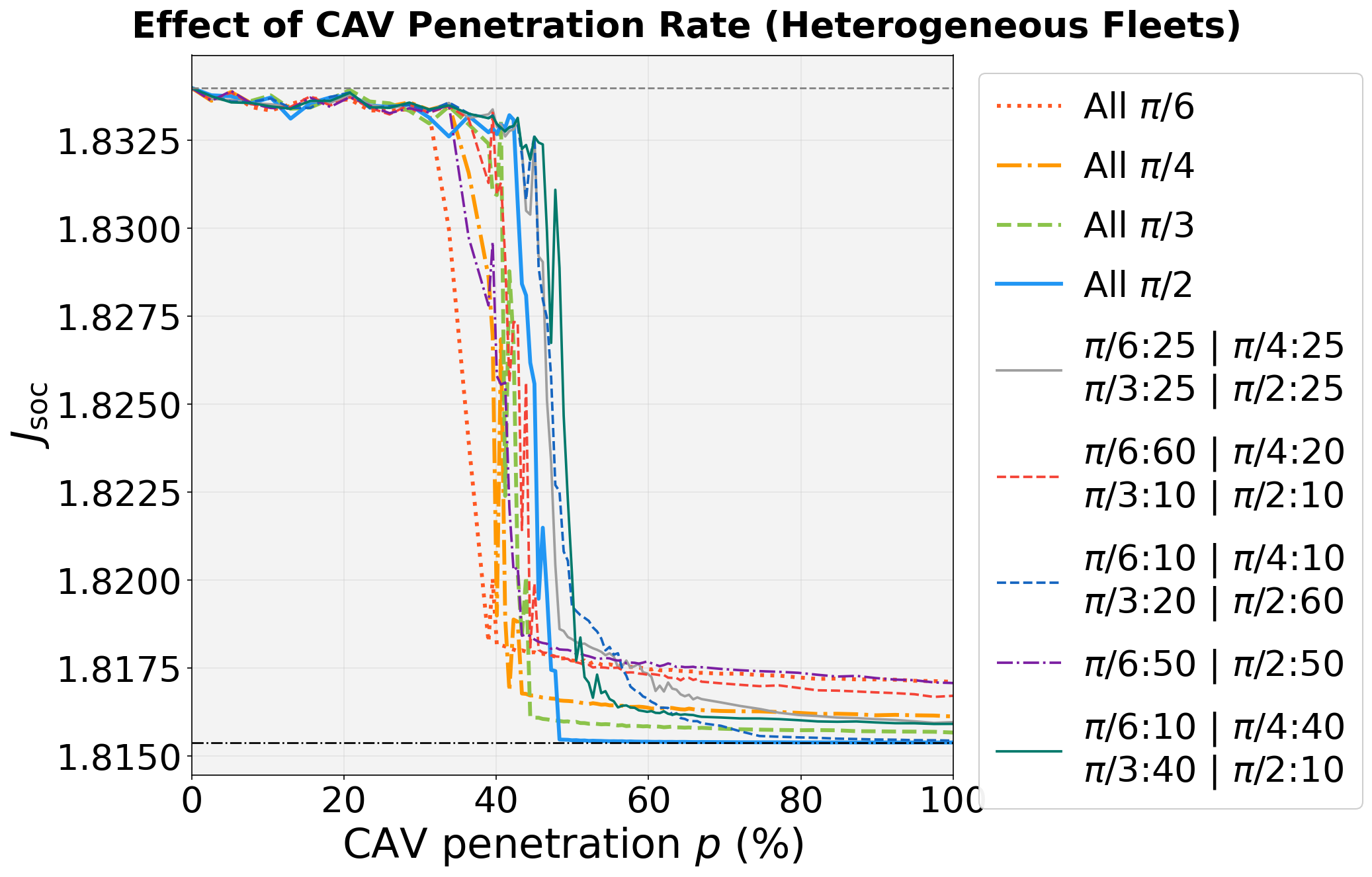}
\caption{Social cost for homogeneous and heterogeneous CAV fleets.
Profiles with matched $\bar\theta_{\mathrm{eff}}$ have similar
equilibria.}
\label{fig:weighted_theta}
\end{figure}

\textbf{Convergence speed and stability.}
Fig.~\ref{fig:teq} reports the settling time $T_{\mathrm{eq}}$ over
$(p,\eta_{\mathrm{CAV}})$ at $\theta=\pi/2$. Low-penetration fleets
require conservative learning rates, whereas higher penetration
supports faster adaptation without sustained oscillation.

\begin{figure}[t]
\centering
\includegraphics[width=0.68\linewidth]{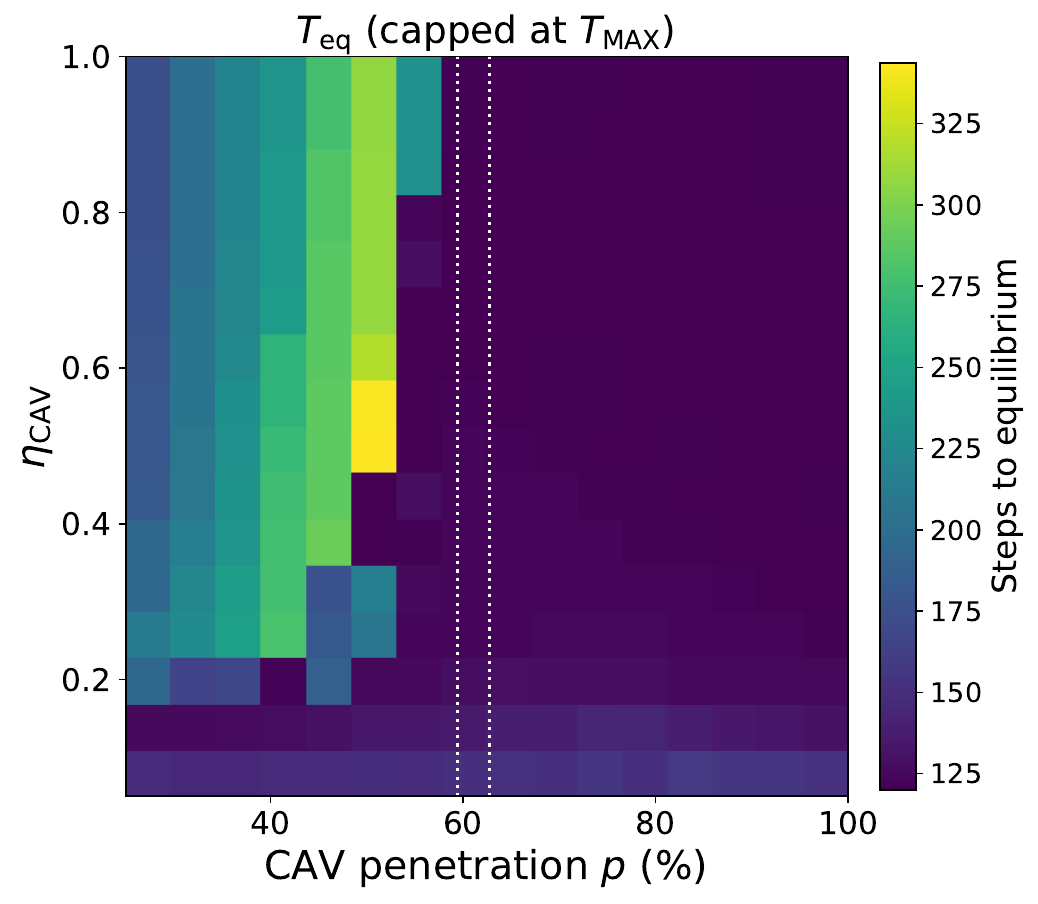}
\caption{Steps to equilibrium over $(p,\eta_{\mathrm{CAV}})$ at
$\theta=\pi/2$; dotted lines indicate $p_1$ and $p_2$.}
\label{fig:teq}
\end{figure}

\textbf{Tracking time-varying demand.}
We vary $n_2(t)$ sinusoidally and compare the ABM trajectory with the
instantaneous quasi-static equilibrium. Fig.~\ref{fig:tracking} shows
that the fleet follows slowly varying demand more accurately: the RMSE
is $\approx$ $0.007$ for period $200$ and $0.012$ for period $40$.

\begin{figure}[t]
\centering
\includegraphics[width=\linewidth]{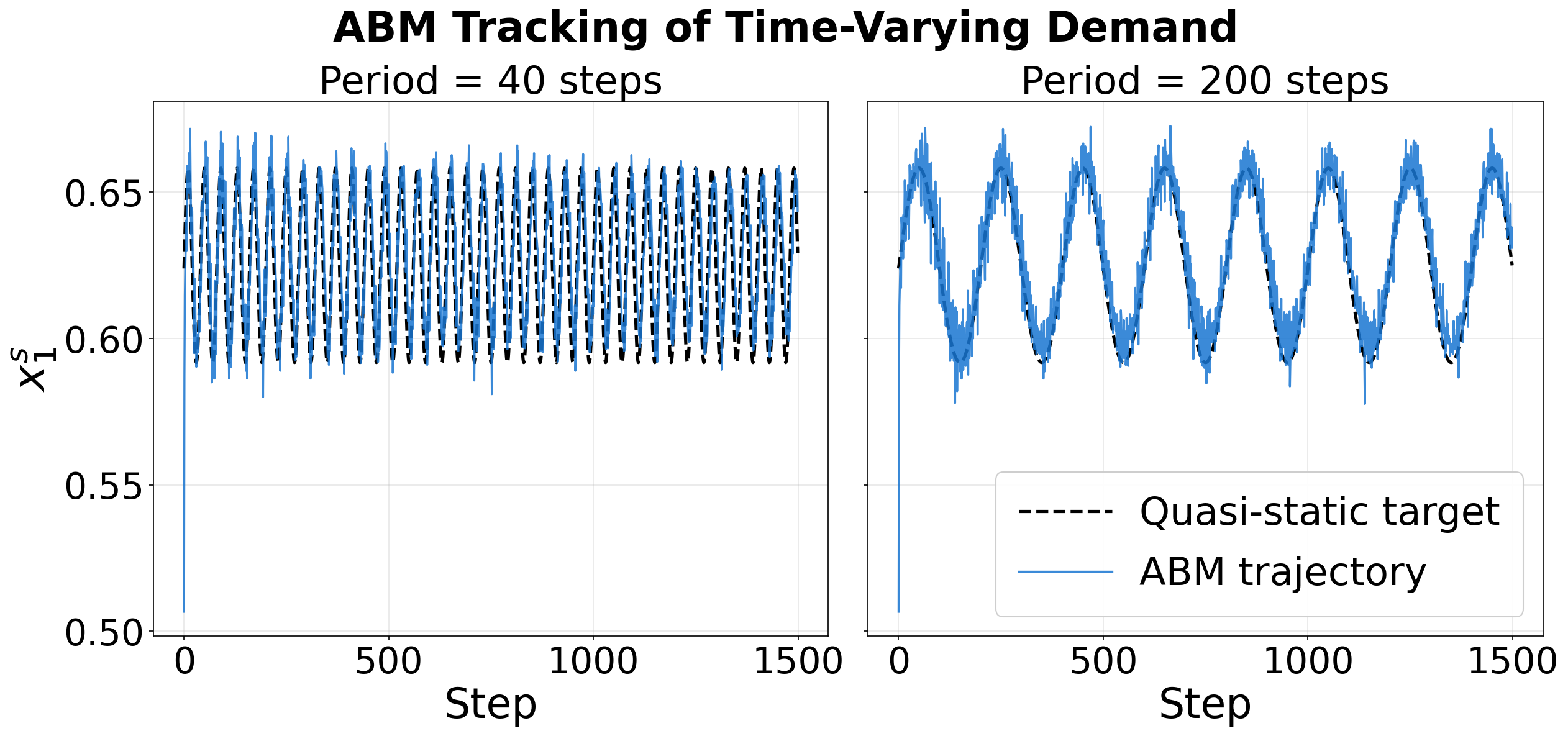}
\caption{ABM tracking of sinusoidal demand with periods $40$ and $200$
for $p>p_2$.}
\label{fig:tracking}
\end{figure}

\textbf{Distributional outcomes and free-riding.}
Aggregate social cost does not reveal how delay is distributed across
vehicles. Fig.~\ref{fig:dist_left} compares per-agent costs across
homogeneous and mixed SVO profiles. The fully altruistic fleet has the
highest measured inequity ($G\approx0.014$), while the mixed profiles
have lower inequity ($G\approx0.009$). Thus, full altruism can
concentrate the cooperative burden on the vehicles that absorb the
conflict-zone delay.

\begin{figure}[t]
\centering
\includegraphics[width=\linewidth]{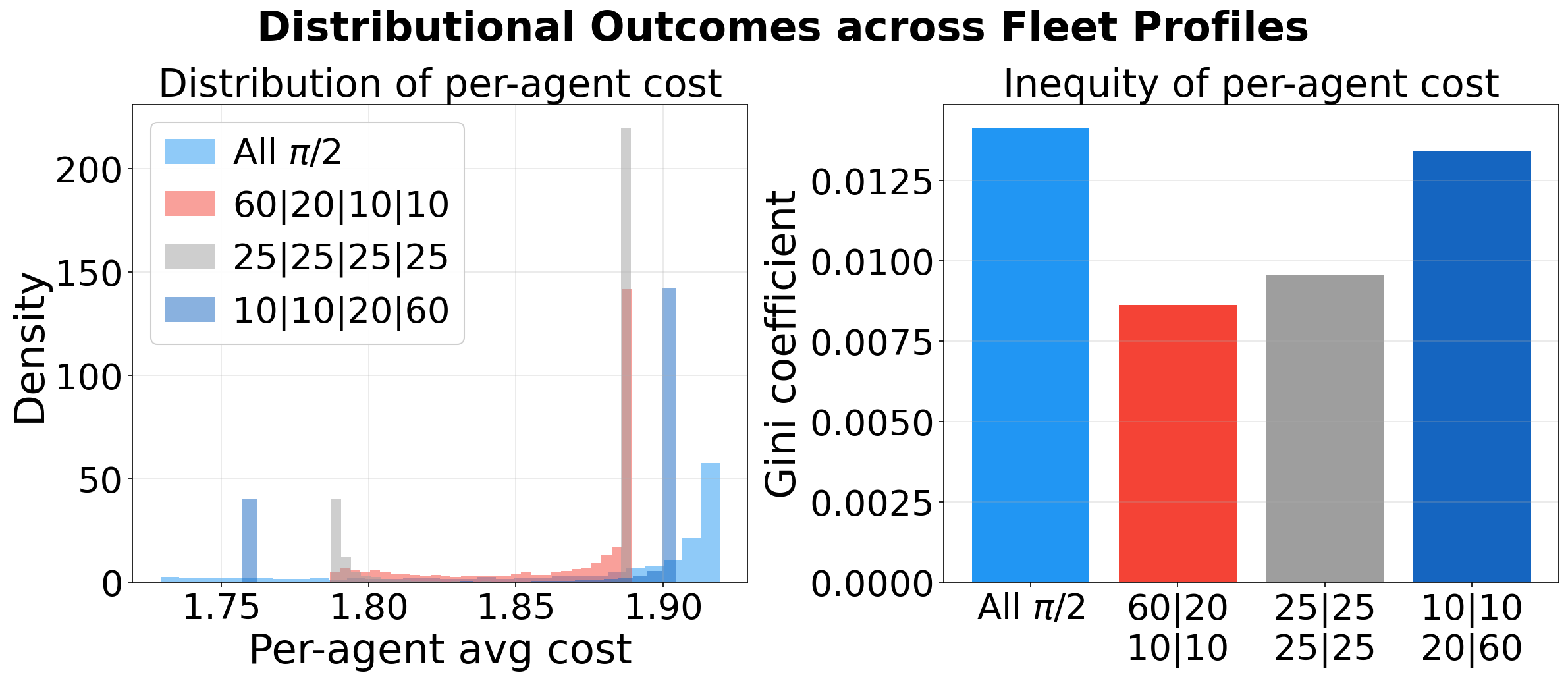}
\caption{Per-agent cost distributions (left) and Gini coefficients
(right) across CAV fleet profiles.}
\label{fig:dist_left}
\end{figure}

\textbf{Partial observability and belief sharing.}
True demand $n$ is replaced by a noisy estimate
$y_{i,t}=n+\varepsilon_{i,t}$, where
$\varepsilon_{i,t}\sim\mathcal N(0,\sigma_{\mathrm{obs}}^2I)$.
Fig.~\ref{fig:partial_jsoc} shows that observation noise delays the
transition between penetration regimes but does not change the
social-optimum floor. Cooperative belief sharing is closer to the
perfect-information result than a memoryless estimator.

\begin{figure}[t]
\centering
\includegraphics[width=0.9\linewidth]{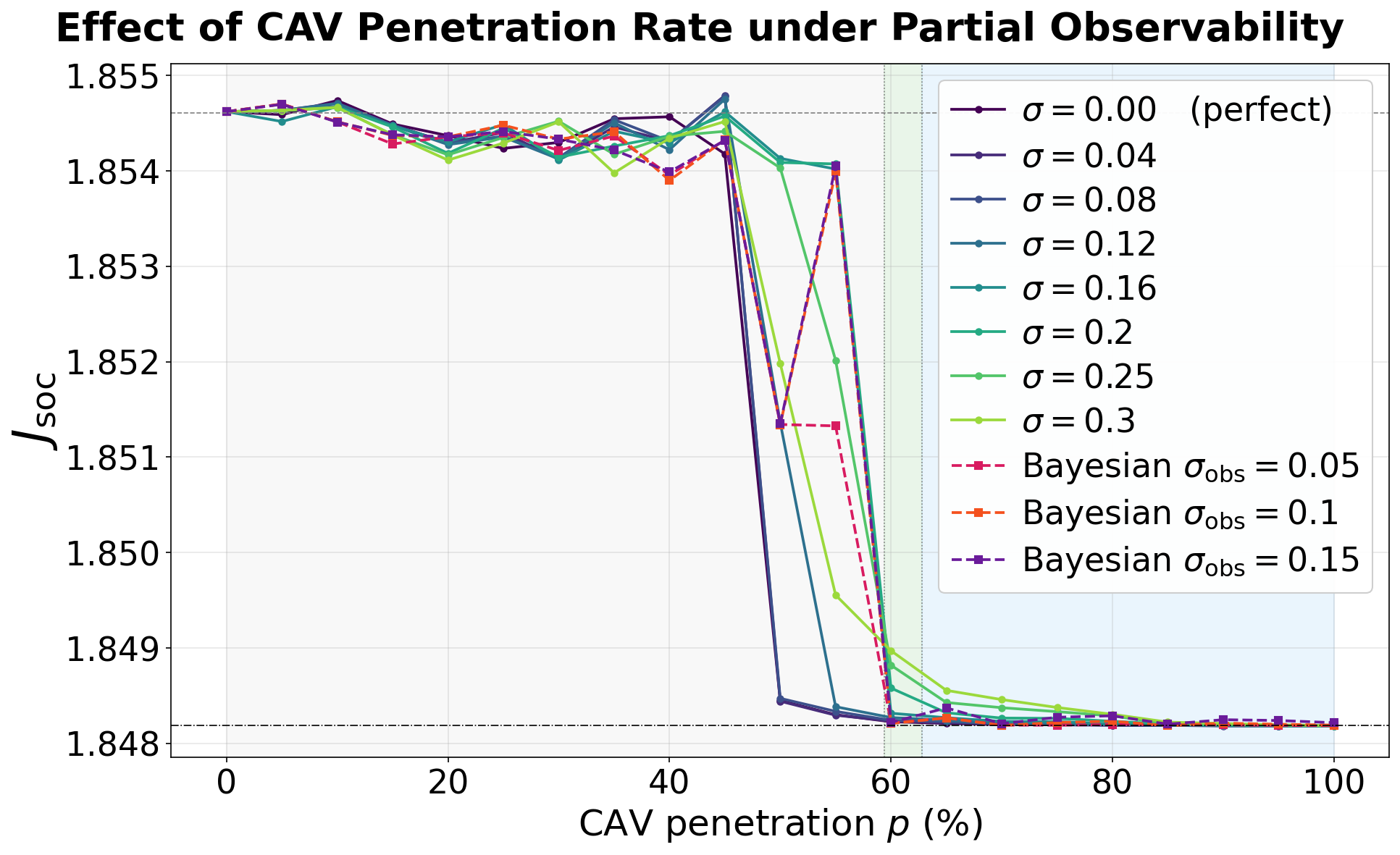}
\caption{Social cost under perfect information, memoryless estimation,
and cooperative belief sharing.}
\label{fig:partial_jsoc}
\end{figure}

\textbf{Recovery from demand shocks.}
We equilibrate the fleet at balanced demand and apply a shock
$n_2:0.33\rightarrow0.60$ at $t=200$. As shown in
Fig.~\ref{fig:shock_recovery}, all fleets reach the same post-shock
equilibrium, but recovery is faster at higher penetration:
$p\ge70\%$ settles within a few steps, whereas $p=40$--$50\%$ requires
approximately $30$--$40$ steps.

\begin{figure}[t]
\centering
\includegraphics[width=0.9\linewidth]{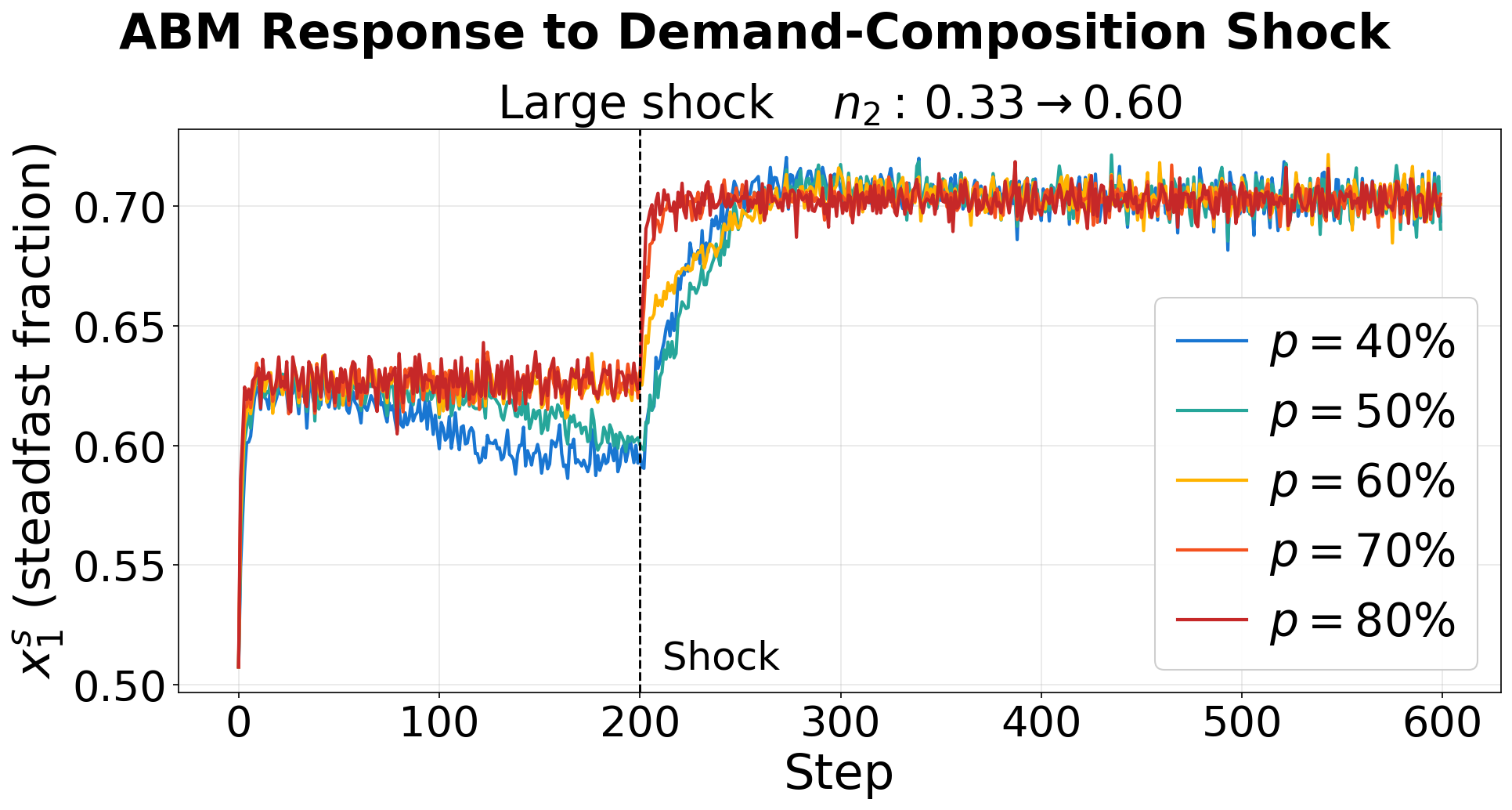}
\caption{Recovery after a demand-composition shock across CAV
penetration rates.}
\label{fig:shock_recovery}
\end{figure}

\textbf{Belief sharing under shocks.}
Bayesian belief sharing reduces estimation variance under stationary
demand but can retain stale information after a shock. We therefore
discount past observations using
\(
\tau_{t+1}
=
\lambda\tau_t+\sigma_{\mathrm{obs}}^{-2},
\ \  \lambda\in(0,1].
\)
Fig.~\ref{fig:forgetting} shows that standard fusion
($\lambda=1$) may require hundreds of steps to recover, whereas
$\lambda\le0.95$ restores convergence within a few steps.

\begin{figure}[t]
\centering
\includegraphics[width=0.9\linewidth]{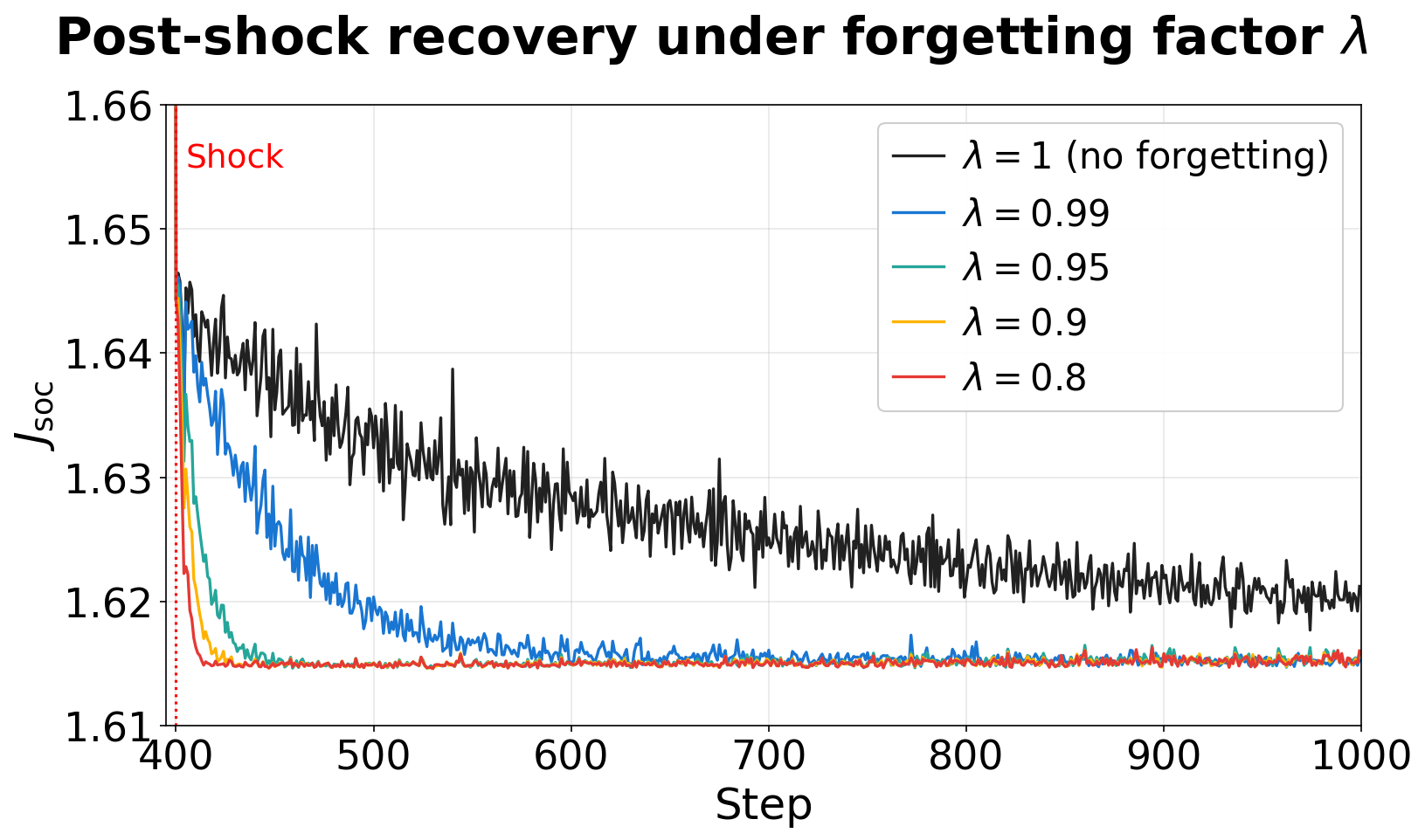}
\caption{Post-shock recovery of social cost under different forgetting
factors $\lambda$; the shock occurs at $t=400$.}
\label{fig:forgetting}
\end{figure}

\section{Conclusions}
We developed an agent-based realization of the macroscopic
weaving-ramp altruism framework, proved its convergence to the
macroscopic equilibrium, and used it to derive deployment-level
findings invisible to the static analysis. These results
are limited to a single weaving ramp under a calibrated affine-cost
model with homogeneous learning rates. Future work includes
networks of coupled weaving sections, and heterogeneous and adaptive
learning dynamics.

\bibliography{refs}

\end{document}